\documentclass[12pt,reqno,letterpaper,oneside]{amsart}
\usepackage{amsmath}
\usepackage{amsfonts}
\usepackage{amssymb}
\usepackage{graphicx}
\usepackage{color}

\renewcommand\epsilon{\varepsilon}

\DeclareMathOperator*\myliminf{\underline{\rm lim}}
\DeclareMathOperator*\mylimsup{\overline{\rm lim}}

\newcommand{\define}{\em}

\newcommand{\R}{{\mathbb R}}
\newcommand{\N}{{\mathbb N}}
\newcommand{\PP}{{\mathbb P}}
\newcommand{\QQ}{{\mathbb Q}}
\newcommand{\EE}{{\mathbb E}}
\newcommand{\FF}{{\mathcal F}}

\newcommand{\NN}{{\mathcal N}}

\renewcommand{\SS}{{\mathcal S}}
\renewcommand{\AA}{{\mathcal A}}
\newcommand{\cd}{\text{c\` adl\` ag} }

\newcommand{\scl}[2]{\langle #1,#2 \rangle} 

\newcommand{\abs}[1]{\left| #1 \right|}
\newcommand{\set}[1]{\left\{#1\right\}} 
\newcommand{\sets}[2]{\set{#1\,:\,#2}} 
\newcommand{\inds}[1]{ {\mathbf 1}_{\set{#1}}} 
\newcommand{\seq}[1]{\set{#1_n}_{n\in\N}}

\newcommand{\eps}{\varepsilon}
\newcommand{\ld}{\lambda}

\newcommand{\conv}{\operatorname{conv}}
\newcommand{\el}{{\mathbb L}} 

\newcommand{\lone}{\el^1}

\newcommand{\linf}{\el^{\infty}}
\newcommand{\linfd}{(\linf)^*} 

\newcommand{\EN}{{\mathcal E}} 
\newcommand{\prf}[1]{ ( #1 )_{t\in [0,T]}} 
\newcommand{\norm}[1]{{||#1||}}

\newcommand{\cii}{(c^i)_{i=1,\dots, d}}

\newcommand{\AF}{\AA^f}
\newcommand{\UF}{{\mathcal U}^f}
\newcommand{\uxi}{\left(\UU^1(c^1),\dots, \UU^d(c^d)\right)}
\newcommand{\bu}{{\mathbf u}}
\newcommand{\bc}{{\mathbf c}}
\newcommand{\br}{{\mathbf r}}
\newcommand{\UU}{\mathbb{U}}

\newcommand{\UFM}{\UF_{-}}
\newcommand{\bb}{\mathbf{b}}

\newcommand{\XHc}{X^{H,c,e}}
\newcommand{\hQ}{\hat{Q}}
\newcommand{\hQQ}{\hat{\QQ}}
\newcommand{\OO}{{\mathcal O}}

\numberwithin{equation}{section}
\theoremstyle{plain}                
\newtheorem{theorem}{Theorem}[section]
\newtheorem{lemma}[theorem]{Lemma}
\newtheorem{proposition}[theorem]{Proposition}

\theoremstyle{definition}           
\newtheorem{definition}[theorem]{Definition}
\newtheorem{example}[theorem]{Example}

\newtheorem{assumption}[theorem]{Assumption}
\theoremstyle{remark}
\newtheorem{remark}[theorem]{Remark}

\title[Financial equilibria]{Financial equilibria in the semimartingale setting:
complete markets and markets with withdrawal constraints}
\author{Gordan \v Zitkovi\' c}
\address{Department of Mathematical Sciences\\ Carnegie Mellon
University \\ Wean Hall 7209 \\ Pittsburgh, PA 15213}
\date{\today}
\begin{document}
\maketitle
\begin{abstract} We establish existence of stochastic financial
equilibria on filtered spaces more general than the ones generated
by finite-dimensional Brownian motions. These equilibria are
expressed in real terms and span complete markets or markets with
withdrawal constraints. We deal with random endowment density
streams which admit jumps and general time-dependent utility
functions on which only regularity assumptions are imposed. As a
side-product of the proof of the main result, we establish a novel
characterization of semimartingale functions.
\end{abstract}

\section{Introduction}
\setcounter{equation}{0}
\paragraph{\em Existing results and history of the problem.}
The existence of financial equilibria in continuous-time financial
markets is one of the central problems in financial theory and
mathematical finance. Unlike the problems of utility maximization
and asset pricing where the price dynamics are given, the
equilibrium problem is concerned with the origin of  security
prices themselves. More precisely, our goal is to construct a
stochastic market with the property that when the price-taking
agents act rationally, supply equals demand. Of course, there are
many ways to interpret the previous sentence, even in the setting
of continuous-time stochastic finance - let alone broader
financial theory or economics as a whole. We are, therefore,
really talking about a whole class of problems.

Before delving into the specifics of our formulation, let us
briefly touch upon the history of the problem. Given the amount of
research published on the various facets of the financial
equilibrium, we can only mention a tiny fraction of the work
leading directly to the present paper. Many seminal contributions
not directly related to our research are left out. The notion of
competitive equilibrium prices as an expression of the basic idea
that the laws of supply and demand determine prices was introduced
by Leon Walras (see \cite{Wal74})  130 years ago. Rigorous
mathematical theory starts with \cite{ArrDeb54}. Continuous-time
stochastic models have been investigated by \cite{DufHua85} and
\cite{Duf86}, among many others. The direct predecessor of this
paper is the work of Karatzas, Lakner, Lehoczky and Shreve in
\cite{KarLakLehShr91}, \cite{KarLehShr90} and \cite{KarLehShr91}.
A convenient exposition of the results of these papers can be
found in Chapter 4. of \cite{KarShr98}. Recently, existence of an
equilibrium functional when utilities exhibit intertemporal
substitution properties has been established in \cite{BanRie01}.
\paragraph{\em Our contributions.} The motive leading our research was
to investigate how the relaxation of the assumption that the
filtration is generated by a Brownian motion affects the existence
theory for the financial equilibrium, and how stringent conditions
on the primitives (utilities, endowments, filtration) one needs to
assume in this case. We were particularly keen to impose minimal
conditions on utility functions and to allow endowment density
processes to admit jumps. As we are primarily concerned with the
{\em existence} of an equilibrium market, we stress that we have
not pursued in any detail the questions of uniqueness or the
financial consequences of our setup. We leave this interesting
line of research for the future, and direct the reader to
\cite{Dan93} and \cite{DanPon92}.  In the following paragraphs we
describe several directions in which this work extends existing
theory.

First, we start from a right-continuous and complete filtration
which we {\em do not} require to be generated by a Brownian
motion. Consequently, we look for the price processes in the set
of all finite-dimensional semimartingales, thus allowing for the
equilibrium prices with jumps. The conditions we impose on the
filtration are directly related with the possibility of obtaining
a {\em finite} number of assets spanning all uncertainty. In this
way, virtually any complete arbitrage-free market known in the
financial literature can arise as an equilibrium in our setting.

Second, we introduce a simple constraint in our model by limiting
the amounts the agents can withdraw from the trading account in
order to finance a consumption plan. This constraint is phrased in
terms of a with\-draw\-al-cap process, which we allow to take
infinite values - effectively including the possibility of a fully
complete market, with no withdrawal cap whatsoever.

Third, we relax regularity requirements imposed on the utility
functions. While these are still stronger than the typical
conditions found in the utility-maximization literature, we show
that one can develop the theory with assumptions less stringent
than, e.g. those in Chapter 4., \cite{KarShr98}. We also deal with
utility functionals which are not necessarily Mackey-continuous
due to unboundedness of the utility functions in the neighborhood
of zero. Moreover, there is no need for fine growth conditions
such as {\em asymptotic elasticity} (see \cite{KraSch99}) in our
setting. A principal feature of our model - jumps in the endowment
density processes - warrants the use and development of tools from
the general theory of stochastic processes. It is in this spirit
that we provide a novel characterization of semimartingale
functions (the functions of both time and space arguments, that
yield semimartingales when applied to semimartingales). Finally, a
result due to M\' emin and Shiryaev (\cite{MemShi79}) is used as
the most important ingredient in establishing a sufficient
condition on a positive semimartingale for the local martingale
part in its multiplicative decomposition to be a true martingale.

Another feature in which this paper differs from the classical
work (e.g. \cite{KarLehShr90}, \cite{KarLakLehShr91}) is in that
we do not introduce the representative agent's utility function
(which is impossible due to withdrawal constraints). Instead we
use Negishi's approach (see \cite{Neg60}) in the version described
in \cite{MasZam91}. This way the proof the existence of a
financial equilibrium is divided into two steps. In the first step
we establish the existence of an equilibrium pricing functional
(an {\em abstract equilibrium}). Next, we implement this pricing
functional through a stochastic market consisting of a finite
number of semimartingale-modeled assets.

\paragraph{\em Organization of the paper and some remarks on the notation.}
After the Introduction, in Section \ref{sec:model} we describe the
model, state the assumptions on its ingredients and pose the
central problem of this work. Section \ref{sec:exabs}  introduces
an abstract setup and establishes the existence of a financial
equilibrium there. In Section \ref{sec:abstosto}, we transform the
abstract equilibrium into a stochastic equilibrium as defined in
Section \ref{sec:model}. Finally, in Appendix \ref{sec:semi} we
develop the semimartingale results used in Section
\ref{sec:abstosto}: characterization of semimartingale functions,
and regularity of multiplicative decompositions. Apart from being
indispensable for the main result of our work, we hope they will
be of independent interest, as well.

Throughout this paper, all stochastic processes will be defined on
the time horizon $[0,T]$, where $T$ is a positive constant. To
relieve the notation, the stochastic process $\prf{X_t}$ will be
simply denoted by $X$, and its left-limit process $(X_{t-})_{t\in
[0,T]}$, by $X_-$. Unless specified otherwise, (in)equalities
between \cd processes will be understood pointwise, modulo
indistinguishability, i.e., $X\leq Y$ will mean $X_t\leq Y_t$, for
all $t\in [0,T]$, a.s. Finally, we use both notations ``$X(t)$''
and ``$X_t$'' interchangeably, the choice depending on
typographical circumstances.

\section{The Model}
\label{sec:model}
\paragraph{\em The information structure.}
We consider a stochastic economy on a finite time horizon $[0,T]$.
The uncertainty reveals itself gradually and is modeled by a
right-continuous and complete filtration $\prf{\FF_t}$ on  a
probability space $(\Omega,\FF,\PP)$, where we assume that
$\FF_0=\set{\emptyset,\Omega}$ mod $\PP$ and $\FF=\FF_T$. In order
for the finite-dimensional stochastic process spanning all the
uncertainty to exist, the size of the filtration $\prf{\FF_t}$
must be restricted:
\begin{definition}
\label{frp} A filtered probability space
$(\Omega,\FF,\prf{\FF_t},\PP)$, with $\prf{\FF_t}$ satisfying the
usual conditions, is said to have the {\bf finite representation
property} if for any probability $\QQ$, equivalent to $\PP$, there
exist a finite number $n$ of $\QQ$-martingales $Y^1,\dots, Y^n$
such that
\begin{enumerate}
    \item $Y^i$ and $Y^j$ are orthogonal for $i\not =j$, i.e.,
    the quadratic covariation $[Y^i,Y^j]_t$ vanishes for all
    $t\in [0,T]$, a.s.
    \item for every bounded $\QQ$-martingale $M$ there exists an
    $n$-dimensional predictable,
    $(Y^1,\dots, Y^n)$-integrable
    process $(H^1,\dots, H^n)$ such that
    \[\text{$M_t=\EE^{\QQ}[M_T]+\sum_{i=1}^n \int_0^t H^i_u\, dY^i_u,\
     \text{for all $t\in [0,T]$, a.s.}$}\]
\end{enumerate}
The smallest such number $n$ is called the {\bf martingale
multiplicity} of $(\Omega,\FF,\prf{\FF_t},\PP)$.
\end{definition}
\begin{example} The filtered probability spaces with finite
representation property include $n$-dimensional Brownian
filtration, filtrations generated by Poisson processes,
filtrations generated by Dritschel-Protter semimartingales (see
\cite{ProDri99}), or combinations of the above.
\end{example}
\begin{remark}
The notion of martingale multiplicity and the related notion of
the {\em spanning number of a filtration} have been introduced by
Duffie in \cite{Duf86}. Definition \ref{frp} differs from Duffie's
in that we explicitly require the existence of martingales
$(Y^1,\dots,Y^n)$, for {\em each} probability measure
$\QQ\sim\PP$. In \cite{Duf85}, Duffie proves that if we only
considered probability measures with $\frac{d\QQ}{d\PP}\in\linf$
in Definition \ref{frp}, it would be enough to postulate the
existence of the processes $(Y^1,\dots,Y^n)$ under $\PP$. It is an
open question whether one can achieve such a simplification under
less stringent conditions on $\QQ$.
\end{remark}
\begin{assumption}[Finite representation property] \label{ftp}
The filtered probability space $(\Omega,\FF,\prf{\FF_t},\PP)$ has
the finite representation property.
\end{assumption}
\begin{remark}
The finite representation property is used to ensure that the
existence of a stochastic implementation of an abstract financial
equilibrium with only a finite number of assets. Without this
property one could still build a financial equilibrium, but the
number of assets needed to span all the uncertainty might be
infinite.
\end{remark}

\paragraph{\em Random endowments.} There are $d\in\N$ agents in our economy
each of whom is receiving a {\bf random endowment} - a bounded and
strictly positive income stream, modeled by a semimartingale
$e^i$. We interpret the random variable $\int_0^t e^i_u\,du$ as
the total income received by agent $i$ on the interval $[0,t]$,
for $t<T$. At time $t=T$ there is a lump endowment of $e^i(T)$. To
simplify the notation, we introduce the measure $\kappa$ on
$[0,T]$ by $d\kappa_t=dt$ on $[0,T)$ and $\kappa(\set{T})=1$. The
cumulative random endowment on $[0,t]$ can now be represented as
$\int_0^t e^i_t\, d\kappa_t$, for all $t\in [0,T]$.
\begin{remark}
The results in this paper can be extended to the case where
$\kappa$ is an optional random measure with $\kappa(\set{T})>0$,
a.s. We do not pursue such an extension, as it would not add to
the content in any interesting way.
\end{remark}

In order for certain stochastic exponentials to be uniformly
integrable martingales, we need to impose a regularity requirement
on $e^i$, $i=1,\dots, d$, described in detail in  Appendix
\ref{sub:mult}.
\begin{definition}\label{def:nnx}
For a special semimartingale $X$, let $\NN(X)=\scl{M}{M}_T$, where
$X=M+A$ is a decomposition of $X$ into a local martingale $M$ and
a predictable process $A$ of finite variation, and $\scl{M}{M}$
denotes the compensator of the quadratic variation $[M,M]$.
\end{definition}
\begin{remark}
The random variable $\NN(X)$ from Definition \ref{def:nnx} will
usually be used in requirements of the form $\NN(X)\in\linf$.
Existence of the compensator $\scl{M}{M}$ and the special
semimartingale property of $X$ are tacitely assumed as parts of
such requirements.
\end{remark}
 The full strength of the following assumption
on random endowment processes $e^i$, $i=1,\dots,d$, is needed for
the existence of a stochastic equilibrium (Theorem \ref{main}),
and only part 1. for the abstract equilibrium (Theorem
\ref{exabs}).
\begin{assumption}[Regularity of random endowments]\label{ends}
For $i=1,\dots,d$,
\begin{enumerate}
\item \label{ends:abs} $e^i$ is an optional process, with
$\eps\leq e^i\leq 1/\eps$, for some $\eps>0$, \item
\label{ends:semi} $e^i$ is a (special) semimartingale and
$\NN(e^i)\in \linf $.
\end{enumerate}
\end{assumption}
\begin{example}
Processes $e^i$ satisfying conditions of Assumption \ref{ends}
include linear combinations of  processes of the form
$Y_t=h(t,X_t)$ where $1/\eps \geq  h \geq\eps>0$ is a
$C^{1,2}$-function, with $h_x$, and $h_{xx}$ uniformly bounded,
and $X$ is
 a diffusion process with a bounded diffusion coefficient, or
a L\' evy process whose jump measure $\nu$ satisfies $\int_{\R}
x^2\, \nu(dx)<\infty$.  Homogeneous and inhomogeneous Poisson
processes and non-exploding continuous-time Markov chains are
examples of allowable processes $X$.
\end{example}
\paragraph{\em Utility functions.} Apart from being characterized
by the random endowment process, each agent represents her
attitude towards risk by a von Neumann-Morgenstern
 utility function $U^i$. Before we list the exact regularity
  assumptions placed on $U^i$,
 we need the following  definition:
\begin{definition}
For a continuously differentiable function $f:[x_1,x_2]\to\R$ we
define the {\bf total convexity norm
$\norm{f}=\norm{f}_{[x_1,x_2]}$} by
\begin{equation}
\nonumber
    \begin{split}
      \norm{f}_{[x_1,x_2]}=\abs{f(x_1)}+
      \abs{f'(x_1)}+{\mathrm TV}(f';\, [x_1,x_2]),
    \end{split}
\end{equation}
where ${\mathrm TV}(f';\, [x_1,x_2])$ denotes the total variation
of the derivative $f'$ of $f$ on $[x_1,x_2]$. A function
$f:[0,T]\times [x_1,x_2]\to\R$, continuously differentiable in the
second variable, is said to be {\bf convexity-Lipschitz} if there
exists a constant $C$ such that, for all $t,s\in [0,T]$, we have
$\norm{f(t,\cdot)-f(s,\cdot)}\leq C \abs{t-s}$.
 A function $f:[0,T]\times I\to\R$ (where I is a subset of $\R$)
   is called {\bf locally convexity-Lipschitz} if
 its restriction $f|_{[0,T]\times [x_1,x_2]}$ is
 convexity-Lipschitz, for any compact interval $[x_1,x_2]$.
\end{definition}
\begin{remark}
A sufficient condition for a function $f:[0,T]\times I\to\R$ to be
(locally) convexity-Lipschitz is that $f(t,\cdot)\in C^2(I)$, for
all $t\in [0,T]$, and $f_{xx}(x,\cdot)$ is Lipschitz, (locally)
uniformly in $x$.
\end{remark}
\begin{assumption}[Regularity of utilities] For each $i=1,\dots, d$,
 the utility function
$U^i:[0,T]\times (0,\infty)\to\R$ has the following  properties
\label{utilreg}
\begin{enumerate}
\item \label{utilreg:util} $U^i(t,\cdot)$ is strictly concave,
continuously differentiable and strictly increasing for each $t\in
[0,T]$. Moreover, the function $U(\cdot,x)$, is bounded for any
$x\in (0,\infty)$. \label{Uconc} \item \label{utilreg:inv} The
{\em inverse-marginal-utility} functions $I^i:[0,T]\times
(0,\infty)\to (0,\infty)$, $I^i(t,y)=U_x(t,\cdot)^{-1}(t,y)$ are
locally convexity-Lipschitz and satisfy
\begin{equation}
\label{inadas}
    \begin{split}
      \lim_{y\to\infty} I^i(t,y)=0,\ \lim_{y\to 0} I^i(t,y)=\infty,\
      \text{uniformly in $t\in [0,T]$.}
    \end{split}
\end{equation}
\end{enumerate}
\end{assumption}
\begin{example} The most important example of a utility function
satisfying Assumption \ref{utilreg} is so-called {\em discounted
utility} $U(t,x)=\exp(-\beta t) \hat{U}(x)$, where $\beta>0$ is
the impatience factor, and $\hat{U}\in C^2(\R_+)$ satisfies
$\hat{U}'>0$ and $\hat{U}''$ is a strictly negative function of
finite variation on compacts. A sufficient (but not necessary)
condition for this is $\hat{U}\in C^3(\R_+)$. Power utilities
$\hat{U}(x)=x^p/p$, for $p\in (-\infty,1)\setminus \set{0}$ and
$\hat{U}(x)=\log(x)$ belong to this class.
\end{example}
\begin{remark}Unlike the problems of utility
maximization (see \cite{KraSch99}, e.g.) where the utility
function is only required to be strictly concave and continuously
differentiable, existence of financial equilibria requires a
higher degree of smoothness (compare to Chapter 4.,
\cite{KarShr98}, where the existence of three continuous
derivatives is postulated in the Brownian setting).
\end{remark}
Total utility accrued by an agent whose consumption equals $c_t
(\omega)$ at time $t\in [0,T]$ in the state of the world
$\omega\in\Omega$, will be modeled as the aggregate of {\em
instantaneous utilities} $U^1(t,c_t(\omega))$ in an additive way.
More precisely, for each agent $i=1,\dots, d$, we define the {\bf
utility functional} $\UU^i$, taking values in $[-\infty, \infty]$.
Its action on an optional process $c$ is given~by
$\UU^i(c)\triangleq \EE[\int_0^T U^i(t,c(t))\,
  d\kappa_t]$ when $\EE[\int_0^T \min(0,U^i(t,c(t)))\,
  d\kappa_t]>-\infty$ and $\UU^i(c)=-\infty$, otherwise.
\begin{remark}
Due to the fact that the final time-point $t=T$ plays a special
role in the definition of the endowment processes $e^i$, one would
like to be able to redefine the agent's utility quite freely
there. Utility functions with virtually no continuity requirements
at $t=T$ are indeed possible to include in our framework, but we
decided not to go through with this in order to keep the
exposition as simple as possible. It will suffice to note that
most of the restrictions involving the time variable placed on the
utility functions in Assumption \ref{utilreg} are there to ensure
that the pricing processes obtained in Theorem \ref{exabs} are
semimartingales and not merely optional processes. All of them
superfluous at $t=T$, since the semimartingale property of a
process $\prf{X_t}$ is preserved if we replace $X_T$ by another
$\FF_T$-measurable random variable.
\end{remark}
\paragraph{\em Investment and consumption.}
The basic premise of equilibrium analysis is that agents engage in
trade with each other in order to improve their utilities. To
facilitate this exchange, a stock market consisting of a finite
number of risky assets $S$, and one riskless asset $B$ is to be
set up. In order to have a meaningful mathematical theory, we
shall require these processes to be semimartingales with respect
to $(\Omega, \FF, \prf{\FF_t}, \PP)$. Moreover, both the riskless
asset $B$ and its left-limit process $B_-$ will be required to be
strictly positive \cd predictable processes of finite variation.

An agent trades in the market by dynamically readjusting the
portion of her wealth kept in various risky, or the riskless
asset. This is achieved by a choice of a portfolio process $H$ (in
an adequate admissibility class to be specified shortly) with the
same number of components as $S$. At the same time, the agent will
accrue utility by choosing the consumption rate according to an
optional consumption process $c$. The components of the process
$H$ stand for the number of shares of each risky asset held in the
portfolio. The trading is financed by borrowing from (or
depositing in) the riskless asset. With that in mind, the equation
governing the dynamics of the wealth $\XHc$ of an agent becomes
\begin{equation}
    \label{wealth}
    \begin{split}
  d\XHc_t=H_t\, dS_t+\frac{(\XHc_{t-}-H_t S_{t-})}{B_{t-}}
  \, d B_t-c(t)\, d\kappa_t+e(t)\, d\kappa_t.
    \end{split}
\end{equation}
We assume that the agent has no initial wealth, i.e., $\XHc_0=0$
(this assumption is in place only to simplify exposition). The net
effect of market involvement of the agent is a redistribution of
wealth across times and states of the world. The income stream $e$
(which would have been the only possibility without the market)
gets swapped for another stream - the consumption process $c$.

There are, invariably, exogenous factors which limit the scope of
the market activity. In this paper we deal with one of the
simplest such limitations - withdrawal constraints. After having
traded for the day (with the net gain of $H_t\,
dS_t+(\XHc_{t-}-H_t S_{t-})/B_{t-}\, d B_t$), and having received
the endowment $e_t\, d\kappa_t$, the agent decides to consume
$c_t\, d\kappa_t$. If this amount is too large, it is likely to be
unavailable for withdrawal from the trading account on a short
notice. Therefore, a cap of $\Gamma^i$ is placed on the amount
agent $i$ can consume at time $t$.  We assume that
$\Gamma^i,i=1,\dots,d$ are  $(0,\infty]$-valued \cd adapted
process satisfying $\Gamma^i>e^i$. We impose no withdrawal
restrictions for $t=T$, effectively requiring $\Gamma^i_T=\infty$
a.s. Moreover, an assumption analogous to Assumption \ref{ends} is
placed on $\Gamma^i$:
\begin{assumption}
\label{with} For each $C>0$, the stochastic process
$\min(\Gamma^i,C)$ is a semimartingale satisfying
$\NN(\min(\Gamma^i,C))\in \linf $.
\end{assumption}
In addition to an abstract,  exogenously given withdrawal-cap
processes, in the following example we describe several other
possibilities.
\begin{example} In all of the following examples, we set
$\Gamma^i_T=\infty$:
\begin{enumerate}
    \item {\em Complete markets}: $\Gamma^i_t=\infty$, $t\in [0,T]$.
    \item {\em Proportional constraints}: For a constant $\gamma>1$,
     $\Gamma^i_t=\gamma e^i_t$, $t\in [0,T)$.
    \item {\em Constant overdraft limit}: for $\delta>0$ we set
    $\Gamma^i_t=e^i_t+\delta$, $t\in [0,T)$.
    \end{enumerate}
\end{example}
\paragraph{\em Market Equilibrium.} Before giving a rigorous definition of
an equilibrium market, we introduce the notion of affordability
for a consumption process $c$. Here we assume that the market
structure (in the form of the withdrawal-cap process $\Gamma$, a
finite-dimensional semimartingale $S$ (risky assets), and a
positive predictable \cd process $B$ of finite variation (riskless
asset)) and the random endowment process $e$ are given.
\begin{definition} \label{afford} An
$(S,B,e,\Gamma)$-{\bf affordable consumption-investment strategy}
is a pair $(H,c)$ of an $S$-integrable predictable {\bf portfolio
process} $H$, and an optional {\bf consumption process} $c\geq 0$
such that
\begin{enumerate}
\item There exists $a\in\R$ such that $a+\int_0^t H_u\, dS_u\geq
0$, for all $t\in [0,T]$, a.s. \item The wealth process
$\prf{X_t}$, as defined in (\ref{wealth}), satisfies $X_T\geq 0$,
a.s. \item The consumption process $c$ satisfies $c_t\leq
\Gamma_t$ for all $t\in [0,T]$, a.s.
\end{enumerate}
\end{definition}
\begin{definition} \label{defequ} A pair $(S,B)$ of a
finite-dimensional semimartingale $S$ and a positive predictable
\cd process $B$ of finite variation
  is said to form an {\bf equilibrium
market} if for each agent $i=1,\dots,d$ here exists an
$(S,B,e^i,\Gamma^i)$-affordable consumption-investment strategy
$(H^i,c^i)$ satisfying the following two conditions:
\begin{enumerate}
\item $\sum_i c^i_t=\sum_i e^i_t$ and $\sum_i H^i_t=0$, for all
$t\in [0,T]$, a.s. \item For each $i$, $c^i$ maximizes the utility
functional $\UU^i(\cdot)$ over all $(S,B,e^i,\Gamma^i)$-affordable
consumption-investment strategies $(H,c)$.
\end{enumerate}
\end{definition}

\section{Existence of an abstract equilibrium}
\label{sec:exabs} In this section we establish the existence of an
abstract version of a market equilibrium. The notion of an
abstract equilibrium encapsulates the tenet that markets in
equilibrium should clear when all agents act rationally. The
full-fledged stochastic market has been abstracted away in favor
of a pricing functional $\QQ$. $\QQ$ will be an element of the
topological dual $\linfd$ of the {\em consumption space} $\linf$,
so that the action $\scl{\QQ}{c}$ of $\QQ$ onto a consumption
process $c$ has the natural interpretation of the price of the
consumption stream $c$. Our setup allows for utility functions
unbounded in the neighborhood of $x=0$ (in order to be able to
deal with the important examples from financial theory). Even
though these utilities follow the philosophy of the von Neumann -
Morgenstern theory, they are {\em not} von Neumann - Morgenstern
utilities in the sense of \cite{Bew72}. In fact, the corresponding
utility functionals are not necessarily Mackey-continuous and thus
the abstract theory pioneered by Truman Bewley and others does not
apply directly to our setting. The structure of our proof of the
existence of an abstract equilibrium follows the skeleton laid out
in \cite{MasZam91}. For that reason we focus on the substantially
novel parts of the proof and only outline the rest. In particular,
we present a detailed proof of closedness of the set of utility
vectors in Lemma \ref{UFclosed}, but merely refer to the
corresponding parts of \cite{MasZam91}  for the results whose
derivation is a more-or-less straightforward modification of
existing results.

\paragraph{\em Functional-analytic setup.} In what follows, $\linf$
will denote the Banach space of ($\kappa\otimes \PP$)-essentially
bounded processes, measurable with respect to the $\sigma$-algebra
$\OO$ of $\prf{\FF_t}$-optional sets. $\linf_+$ will denote the
positive orthant of $\linf$, i.e.,  the set of all
$(\kappa\otimes\PP)$-a.e. nonnegative elements in $\linf$. All
$\prf{\FF_t}$-optional processes will be identified with the
corresponding $\OO$-measurable random variables without explicit
mention, and the equalities and inequalities will always be
understood in $(\kappa\otimes\PP)$-a.e. sense.

The set of all bounded consumption processes $c$ satisfying the
consumption constraints introduced via cap processes $\Gamma^i$,
will be denoted by $\AA^i$, i.e., $\AA^i = \sets{c\in\linf_+}{
c\leq \Gamma^i}$. Also, define $\AA=\sets{\cii}{c^i\in\AA^i}$, and
its subset $\AF$ consisting of only those allocations which can be
produced by redistributing the aggregate endowment $e=\sum_i e^i$,
i.e., $\AF=\{\cii\in\AA\,:\, \sum\nolimits_i c^i=e\}$

The topological dual $\linfd$ of $\linf$ can be identified with
the set of all finitely-additive measures $\QQ$ on the
$\sigma$-algebra $\OO$, weakly-absolutely continuous with respect
to $\kappa\otimes\PP$, i.e. for $A\in\OO$, $\QQ[A]=0$ whenever
$(\kappa\otimes\PP)[A]=0$.

\begin{remark}\ \label{Ala}\
We will consider the set of finitely-additive probabilities as a
subset of $\linfd$, supplied with the weak * topology
$\sigma(\linfd, \linf)$. It is a consequence of Alaoglu's theorem
 that any collection of finitely-additive probabilities is relatively
$\sigma(\linfd,\linf)$-compact. Furthermore, the closedness of the
set of finitely-additive probabilities (in the space of all
finite-additive measures, and w.r.t the $\sigma(\linfd,
\linf)$-topology) implies that the cluster-points of nets of
finitely-additive probabilities are finitely-additive
probabilities themselves. In the sequel, weak * topology will
always refer to the $\sigma(\linfd, \linf)$ topology of the pair
$(\linfd,\linf)$.
\end{remark}

We can now define the concept of an abstract equilibrium. Instead
of a semimartingale price process, an abstract equilibrium
requires the existence of a finitely-additive probability
$\QQ\in\linfd$ which takes the role of a pricing functional acting
directly on consumption processes. Given such a finitely-additive
probability $\QQ$, the {\bf budget set} $B^i(\QQ)$ of agent $i$ is
defined by $ B^i(\QQ)=\sets{c\in \linf_+}{ c\in\AA^i \text{\ and\
} \scl{\QQ}{c}\leq\scl{\QQ}{e^i}}$.
\begin{definition}\label{absequ} A pair
$(\QQ,\cii)$ of a finitely-additive probability $\QQ$ and an
allocation $\cii\in\AA$ is called an {\bf abstract equilibrium} if
\begin{enumerate}
\item $\sum_i c^i=\sum_i e^i$,\  i.e., $\cii\in \AF$. \item For
any $i=1,\dots,d$, $c^i\in B^i(\QQ)$ and $\UU^i(c^i)\geq \UU^i(c)$
for all $c\in B^i(\QQ)$.
\end{enumerate}
\end{definition}
\paragraph{\em Existence of an abstract equilibrium.}
To simplify notation in some proofs and statements we assume that
the utility functionals $\UU^i$ are normalized so that
$\UU^i(e^i)=0$ for all $i=1,\dots, d$.

We start by introducing $\UF$ - the set of all $d$-tuples of
utilities which can be achieved by different allocations
$\cii\in\AF$, i.e.,
\begin{equation}
    \label{defu}
    \begin{split}
        \UF=\sets{\uxi}{\cii\in\AF},
    \end{split}
\end{equation}
and $\UFM=\UF-[0,\infty)^d$ -  the set of all vectors in $\R^d$
dominated by some element in $\UF$. The elements in $\UFM$ will be
called {\bf utility vectors}.
 Our first lemma identifies several properties of $\UFM$, the most
 important of which is closedness.
\begin{lemma}\label{UFclosed}
    The set $\UFM$ is non-empty, convex and  closed.
\end{lemma}
\begin{proof}
$\UFM$ is obviously non-empty, and its convexity follows easily
from convexity of $\AF$. It remains to show that it is closed. Let
$\seq{\bu}$, $\bu_n = ( u^1_n, u^2_n, \ldots, u^d_n )$, be a
sequence in $\UFM$
  converging to $\bu = ( u^1, u^2, \ldots, u^d ) \in \R^d$. By the
  definition of the set $\UFM$, there exist two sequences
 $\bc_n=(c^1_n,$ $c^2_n,\dots c^d_n) \in \AF$ and $\br_n=(r^1_n,\dots,
r^d_n)\in\R^d_+$  such that $\UU^i(c^i_n)=u^i_n+r^i_n$. Since
$u^i_n\leq\UU^i(c^i_n)\leq \UU^i(e)<\infty$, we can assume -
passing to a subsequence if necessary - that there exists a vector
$\hat{\bu}=(\hat{u}^1,\dots,\hat{u}^d )$ such that
$\UU^i(c^i_n)\to \hat{u}^i\geq u^i$.

   For any $i=1\dots d$, the
  sequence  $\seq{c^i}$ is bounded in
  $\linf$, and therefore also in $\lone(\kappa\otimes\PP)$.
  By a simple extension of the classical
  Komlos' theorem (see \cite{Sch86}) to the case of
  $\R^d$-valued random variables,  there exists an infinite
   array of nonnegative weights
  $(\alpha_k^n)^{n\in\N}_{k=n,\ldots, k_n}$ and a $d$-tuple
   $\cii$ of
  nonnegative optional processes with the following properties:
  $\sum_{k=n}^{k_n} \alpha_k^n=1$ and $\tilde{c}^i_n=
  \sum_{k=n}^{k_n} \alpha_k^n c^i_k\,\to\, c^i$,
  $(\kappa\otimes\PP)$-a.e.
Consequently, $\sum_i c^i=e$ and $c^i\leq \Gamma^i$, so that
$\cii\in\AA^f$.

To show that $\bu\in\UFM$, we use concavity and right-continuity
of the utility functions and the Fatou Lemma
 (the use of which is justified by the fact that
 $\tilde{c}^i_n\leq e$, for all $i$ and all $n\in\N$)
 in the following chain of inequalities:
\begin{equation}
    \nonumber
    \begin{split}
  \UU^i(c^i)=\UU^i(\lim_n \tilde{c}^i_n)\geq \mylimsup_n
  \UU^i(\tilde{c}^i_n)\geq
 \mylimsup_n \sum_{k=n}^{k_n} \alpha_k^n
 \UU^i(c^i_k)=\lim_n \UU^i(c^i_n)=\hat{u}^i\geq u^i,
    \end{split}
\end{equation}
\end{proof}

The next task is to establish the  existence of {\em supporting
measures} for {\em weakly optimal} utility vectors . We start with
definitions of these two concepts.
\begin{definition}
  A finitely-additive probability $\QQ$ is said
  to \textbf{support} a \label{supp} vector $\bu=(u^1,\dots, u^d)$
   $\in\R^d$ if
  for any allocation $\bc=\cii\in\AA$ with the property that
  $\UU^i ( c^i ) \geq u^i$ for all $i=1,\dots, d$, we have
$    \scl{\QQ}{\sum\nolimits_i c^i} \geq \scl{\QQ}{\sum\nolimits_i
e^i}$.
  The set of all finitely-additive probability measures
  supporting a vector
  $\bu\in\R^d$ is denoted by $P(\bu)$.
\end{definition}

\begin{definition}\label{wopt}
A vector $\bu=(u^1,\dots, u^d)$ in $\UFM$ is said to be {\define
weakly optimal} if there is no allocation $\cii\in\AF$ with the
property that $\UU^i(c^i)>u^i$ for all $i=1,\dots,d$.
\end{definition}
\begin{lemma}[Second Fundamental Theorem of Welfare Economics]
\label{Pnonempty}
  For a weakly optimal utility vector $\bu\in\UFM$, the set
  $P(\bu)$ of finitely-additive probabilities supporting $\bu$ is
  non-empty, convex and weak * compact
\end{lemma}
\begin{proof} The proof relies on a well-know
separating-hyperplane-type argument. See \cite{MasZam91}, Section
8., pp. 1859-1860 for more details.
\end{proof}
Having established the closedness and convexity of the set $\UFM$
in Lemma \ref{UFclosed}, and the existence of supporting
functionals for weakly optimal utility vectors in Lemma
\ref{Pnonempty}, it suffices to use the proof of Theorem 7.1, p.
1856 in \cite{MasZam91} to establish the following abstract
existence theorem:

\begin{theorem}\label{exabs}Under Assumptions
\ref{ends}.\ref{ends:abs}, \ref{utilreg}.\ref{utilreg:util} and
\ref{utilreg}.\ref{utilreg:inv}, there exists an abstract
equilibrium $(\QQ,\cii)$.
\end{theorem}

\section{From abstract to stochastic equilibria}
\label{sec:abstosto} Our next task is to show that the abstract
equilibrium obtained in the previous section  can be implemented
as a stochastic equilibrium. We first note that the equilibrium
functional $\QQ$ must be countably-additive and equivalent to
$\kappa\otimes \PP$. We omit the proof as it follows the argument
from Theorem 8.2, p. 1863 in \cite{MasZam91}, using the fact that
$\Gamma^i>e^i$ and $\Gamma^i_T=\infty$ for all $i=1,\dots, d$.

\begin{lemma} \label{countadd} Let $(\QQ,\cii)$ be an abstract
equilibrium. Then $\QQ$ is countably additive and equivalent to
$\kappa\otimes\PP$.
\end{lemma}
In Lemma \ref{hasform} we use convex duality to describe the
solutions of agents' utility-maximization problems in an
equilibrium:
\begin{lemma} \label{hasform}
Suppose that $(\QQ,\cii)$ is an abstract equilibrium. Then there
exist constants $\ld^i>0$, $i=1,\dots, d$,  such that the
consumption processes $c^i,i=1,\dots,d$ are of the form
\begin{equation}
    \label{formc}
    \begin{split}
         c^i_t=\min(\Gamma^i_t, I^i(t,\ld^i Q_t)),
    \end{split}
\end{equation}
 where $Q=\prf{Q_t}$ is the optional version of the Radon-Nikodym
  derivative of
$\QQ$ with respect to $\kappa\otimes\PP$.
\end{lemma}
\begin{proof}
We prove the lemma for $i=1$.  Let $N(c^1)$ be the set of all
$c\in\linf_+$ such that $c\leq \min(\Gamma^1,\norm{c^1}_{\linf})$.
$N(c^1)$ is a $\sigma(\linf,\lone)$-compact subset in $\linf_+$,
and  by Komlos' Lemma the restriction of $\UU^1$ to $N(c^1)$ is
$\sigma(\linf,\lone)$-upper-semicontinuous and concave. By Lemma
\ref{countadd}, the finitely-additive measure
 $\QQ$ is countably-additive so the
Lagrangean function $L:N(c^1)\times [0,\infty)\to
[-\infty,\infty)$, $L(c,\ld)=\UU^1(c)-\ld \scl{\QQ}{c-e^1}$
satisfies the conditions of the Minimax theorem (see
\cite{Sio58}).We know that the maximizer $c^1$ of the functional
$\UU^1$ over $B^1(\QQ)$ trivially satisfies $c^1\leq
\norm{c^1}_{\linf}$, so
\begin{equation}
    \nonumber
    \begin{split}
\UU^1(c^1)&=\sup_{c\in B^1(\QQ)\cap N(c^1)}\UU^1(c) =\sup_{c\in
N(c^1)} \inf_{\ld\geq
0} L(c,\ld) = \inf_{\ld\geq 0} \sup_{c\in N(c^1)} L(c,\ld)\\
&=\inf_{\ld\geq 0}  \Big( \ld \scl{\QQ}{e^1}+\EE\int_0^T
V(t,\ld Q_t;\, m^1_t)\,d\kappa_t\Big),\\
    \end{split}
\end{equation}
where $m^1_t=\min(\Gamma^1_t,\norm{c^1}_{\linf})$, and the
function $V:[0,T]\times[0,\infty)\times (0,\infty)\to\R$ is given
by
\begin{equation}
    \nonumber
    \begin{split}
  V(t,\ld;\, \xi )\triangleq\sup_{x\in [0,\xi)} (U^1(t,x)-x\ld)
=\begin{cases} V(t,\ld;\,\infty),& \ld>U^1_x(t,\xi), \\
U^1(t,\xi)-\ld\xi,& \ld\leq U^1_x(t,\xi).\end{cases}
    \end{split}
\end{equation}
 $V$ is
convex and nonincreasing in $\ld$, and nondecreasing in $\xi$. The
function $v:[0,\infty)\to [-\infty,\infty]$,  where $v(\ld)=\ld
\scl{\QQ}{e^1}+\EE\int_0^T V(t,\ld Q_t;\, m^1_t)\,d\kappa_t$, is
convex and proper, since $\inf_{\ld\geq 0} v(\ld)=\UU^1(c^1)\in
(-\infty,\infty)$. Furthermore, Assumption
\ref{utilreg}.\ref{utilreg:util} implies the inequality
$V(t,\ld;\, m^1_t)\leq U^1(t,\norm{c^1}_{\linf})$ and the
existence of a constant $D>0$ such that $\UU^1(c^1)\leq v(\ld)\leq
\ld\scl{\QQ}{e^1}+D$, for all $\ld>0$.

Assumption \ref{utilreg}.\ref{utilreg:inv}  ensures the existence
of a constant $C>\norm{c^1}_{\linf}$ such that
$I^1(t,C)<\frac{1}{2} \scl{\QQ}{e^1}$ for all $t\in [0,T]$. Then,
for all $\xi,\ld,\ld_0>0$ with the property that
$\ld>\ld_0>\max(C,U^1_x(t,\xi))$,  we have
\begin{equation}
    \nonumber
    \begin{split}
  V(t,\ld; \xi)& \geq V(t,\ld_0;\xi)+(\ld-\ld_0)
  V_{\ld}(t,\ld_0;\xi) \\
& = V(t,\ld_0; \infty)-(\ld-\ld_0) I^1(t,\ld_0) \geq V(t,\ld_0;
\infty)-\frac{1}{2}\scl{\QQ}{e^1}(\ld-\ld_0).
    \end{split}
\end{equation}
Therefore, if we let $L=  \myliminf_{\ld\to\infty}
\big(\frac{v(\ld)}{\ld}-\scl{\QQ}{e^1}\big)\in [-\infty,\infty]$,
we have
\begin{equation}
    \nonumber
    \begin{split}
L&=\myliminf_{\ld\to\infty}
  \frac{1}{\ld}\EE\int_0^T{V(t,\ld Q_t; m^1_t)}d\kappa_t
   \geq \myliminf_{\ld\to\infty}
  \frac{1}{\ld} \EE\int_0^T{V(t,\ld Q_t; m^1_t)
  \inds{Q_t> \frac{C}{\ld_0}}}d\kappa_t\\ & \hspace{2em} +
  \myliminf_{\ld\to\infty}
  \frac{1}{\ld} \EE\int_0^T{V(t,\ld Q_t; m^1_t)
  \inds{Q_t\leq  \frac{C}{\ld_0}}}d\kappa_t\\
&\geq \myliminf_{\ld\to\infty} \Big(
\frac{1}{\ld}\EE\int_0^T{V(t,\ld_0 Q_t;\, m^1_t)\inds{Q_t>
\frac{C}{\ld_0}}d\kappa_t-
\frac{1}{2\ld}\scl{\QQ}{e^1}(\ld-\ld_0)}\Big)\\
&\hspace{2em} + \myliminf_{\ld\to\infty} \frac{1}{\ld}
\EE\int_0^T{V(t,\ld Q_t; m^1_t)\inds{Q_t\leq
\frac{C}{\ld_0} }}d\kappa_t\\
&\geq -\frac{1}{2}\scl{\QQ}{e^1}+\myliminf_{\ld\to\infty}
\frac{1}{\ld}\EE\int_0^T{V(t,\ld Q_t; m^1_t)\inds{Q_t\leq
\frac{C}{\ld_0} }}d\kappa_t\geq -\frac{1}{2}\scl{\QQ}{e^1}.
    \end{split}
\end{equation}
Hence, $\lim_{\ld\to\infty} v(\ld)=\infty$ and there exists a
constant $\ld^1 \in [0,\infty)$ such that $v(\ld^1)=\UU^1(c^1)$,
i.e.,
\begin{equation}
    \nonumber
    \begin{split}
\EE\int_0^T U^1(t,c^1(t))\, d\kappa_t& =\EE\int_0^T\ld^1 Q_t
e^1_t\, d\kappa_t+\EE\int_0^T V(t,\ld^1
Q_t;\, m^1_t)\,d\kappa_t\\
&\geq \EE\int_0^T\ld^1 Q_t c^1_t\, d\kappa_t+\EE\int_0^T V(t,\ld^1
Q_t;\, m^1_t)\,d\kappa_t.
    \end{split}
\end{equation}
On the other hand, $U^1(t,x)\leq \ld^1 Q_t x+V(t,\ld^1 Q_t
;m^1_t)$ for all $t\in [0,T]$ and $x\in [0,m^1_t]$ (with equality
only for $x=\min(m^1_t,I^1(t,\ld^1 Q_t))$), so $c^1$ must be of
the form (\ref{formc}). To rule out the possibility $\ld^1=0$,
note that it would force $c^1=\Gamma^1$ and   violate the budget
constraint since $\Gamma^1>e^1$.
\end{proof}
\begin{proposition}
The process $Q$ has a modification which is a semimartingale, and
there exists a constant $\eps>0$ such that $\eps\leq Q \leq
1/\eps$.
\end{proposition}
\begin{proof}
By Lemma \ref{hasform} there exists constants $\ld^i >  0$ such
that $e_t= \sum_i c^i_t=\sum_i \min(\Gamma^i_t, I^i(t,\ld^i
        Q_t))$, $\kappa\otimes\PP$-a.s.
Since $(\kappa\otimes\PP)[\sum_i \Gamma^i>e]=1$, we have
$e_t=\min_{\bb\in B} \big( \sum_i b_i I^i(t,\ld^i Q_t)+\sum_i
(1-b_i) \Gamma^i_t\big)$ ,where $B=\set{0,1}^d\setminus
\set{0,\dots,0}$.

For $\bb\in B$, the function $I^{\bb}$, defined by
$I^{\bb}(t,y)=\sum_i b_i I^i(t,\ld^i y)$, is strictly decreasing
in its second argument and shares the properties in Assumption
\ref{utilreg} with each $I^i$. Therefore, there exists a function
$J^{\bb}: [0,T]\times (0,\infty)\to (0,\infty)$  such that
$I^{\bb}(t,J^{\bb}(t,x)))=x$, for all $(t,x)\in [0,T]\times
(0,\infty)$. Thus, with $\Gamma^{\bb}_t=\sum_i (1-b_i)
\Gamma^i_t$, we have
\begin{equation}
    \nonumber
    \begin{split}
        e_t \geq x\ \Leftrightarrow
        I^{\bb}(t,Q_t)+\Gamma^{\bb}_t\geq x,\,\forall\,\bb\in B
        \Leftrightarrow
        Q_t \leq J^{\bb}(t,x-\Gamma^{\bb}_t),\,\forall\,\bb\in B,
    \end{split}
\end{equation}
with $J(t,x-\Gamma^{\bb}_t)=\infty$ for $x\leq \Gamma^{\bb}_t$.
Consequently, $Q_t=\min_{\bb\in B} J^{\bb}(t,e_t-\Gamma^{\bb}_t)$.
 Knowing that the semimartingale property is preserved under
 maximization,
it will be enough to prove that for each $\bb\in B$, $J^{\bb}$ is
a semimartingale function (see Definition \ref{defsemi}). By Inada
conditions (\ref{inadas}) - holding uniformly in $t\in [0,T]$ -
$I^{\bb}$ maps compact sets of the form $[0,T]\times [y_1,y_2]$
into compact intervals. The function $I^{\bb}$ is locally
convexity-Lipschitz, so the conclusion that $Q$ is a
semimartingale follows from Proposition \ref{invreg}.

To show boundedness, we first set  $\bb_1=(1,\dots,1)$ to conclude
that $Q_t\leq
J^{\bb_1}(t,e_t-\Gamma^{\bb_1}_t)=J^{\bb_1}(t,e_t)\in\linf$. On
the other hand, $Q_t=\min_{\bb\in B}
J^{\bb}(t,e_t-\Gamma^{\bb}_t)\geq \min_{\bb\in B} J^{\bb}(t,e_t)$
- a positive quantity, uniformly bounded from below. Therefore,
the semimartingale $Q_t$ is positive and uniformly bounded from
above and away from zero.

\end{proof}

\begin{proposition} \label{hasdec}
The process $Q$ admits a multiplicative decomposition $Q=\hat{Q}
\beta$ where $\hat{Q}$ is a strictly positive uniformly integrable
martingale, and $\beta$ is a strictly positive \cd predictable
process of finite variation.
\end{proposition}
\begin{proof}
By the representation $Q_t=\min_{\bb\in B}
J^{\bb}(t,e_t-\Gamma^{\bb}_t)$, and boundedness of $Q$ from above,
there exists a constant $C>0$ such that $Q_t=\min_{\bb\in B}
J^{\bb}(t,\max(C,e_t-\Gamma^{\bb}_t))$. Propositions
\ref{regissemi}, \ref{stab} and \ref{Decomp}
 complete the proof.
\end{proof}

\paragraph{\em Construction of the equilibrium market.}
Thanks to  Proposition \ref{hasdec}, there exists a measure $\hQQ$
(with $\frac{d\hQQ}{d\PP}=\frac{\hat{Q}_T}{\EE[\hat{Q}_T]}$)
equivalent to $\PP$ such that
\begin{equation}
    \label{defhq}
    \begin{split}
        Q_t=\EE[\frac{d\hQQ}{d\PP}|\FF_t] \beta_t,\ \text{and}\
        \scl{\QQ}{c}=\EE\int_0^T Q_u c_u \, d\kappa_u
        = \EE^{\hQQ}\int_0^T c_u\beta_u\, d\kappa_u.
    \end{split}
\end{equation}
In words, the action of the pricing functional $Q$ on a
consumption stream $c$ can be represented as a $\hQQ$-expectation
of a discounted version $c(u)\beta_u$ of $c$.

Let $n\in\N$ be the martingale multiplicity of the filtration
$\prf{\FF_t}$ under $\hQQ$, and let $(Y_1,\dots,Y_n)$ be an
$n$-dimensional positive $\hQQ$-martingale described in Definition
\ref{frp}. Define the riskless asset $B$ and the stock price
process $S=(S_1,\dots, S_n)$ as follows
\begin{equation}
    \label{defeq}
    \begin{split}
  B(t)=1/\beta(t), S_j(t)=B(t) Y_j(t),\ t\in [0,T],\, j=1,\dots, n.
    \end{split}
\end{equation}
\vspace{-4ex}
\begin{lemma} The pair $(S, B)$, defined in (\ref{defeq}) is an
equilibrium market.
\end{lemma}
\begin{proof}
Let $(\QQ,\cii)$ be the abstract equilibrium which produced
$(S,B)$, and let the measure $\hQQ$ be as in (\ref{defhq}). For
$i=1,\dots, d$, define the $\hQQ$-martingale $\tilde{X}^i$ by
$\tilde{X}^i_t=\EE^{\hQQ}[\int_0^T (c^i(u)-e^i(u))\beta_u \,
du|\FF_t]$. By the finite representation property (Assumption
\ref{ftp}), for each $i=1,\dots, d$ there exists an $S$-integrable
portfolio process $H^i$ such that
$\tilde{X}^i_t=\tilde{X}^i_0+\int_0^t \tilde{H}^i\, dS_u$.
Moreover, the boundedness of processes $c^i$ and $e^i$ guarantees
that $\tilde{H}^i$ satisfies part 1. of Definition \ref{afford}.
Standard calculations involving integration by parts and using the
fact that $B$ is a predictable process of finite variation imply
that the wealth process $X^{\tilde{H}^i,c^i,e^i}$ defined as in
(\ref{wealth}) is bounded and satisfies
$X^{\tilde{H}^i,c^i,e^i}_T\geq 0$. Therefore, $(\tilde{H}^i,c^i)$
is an affordable consumption-investment strategy (as described in
Definition \ref{afford}).

Since $\sum_{i} \tilde{X}^i=0$, the mutual orthogonality of the
$\hQQ$-martingales $Y_1,\dots, Y_n$ implies that $\sum_{i=1}^d
\tilde{H}^i_j(t)=0,\ d[Y_j,Y_j]_t-\text{a.e.}$, for all $j$. In
order to have markets clear for {\em every} $t\in [0,T]$, we
define the portfolio process $H^i=(H^i_1,\dots, H^i_n)$ by
$H^i_j(t)=\tilde{H}^i_j(t)\inds{\sum_i H^i_j(t) = 0}$, for each
$i=1,\dots, d$, so that
\begin{enumerate}
\item $H^i_j(t)=\tilde{H}^i_j(t)$, $d[Y_j,Y_j]$-a.e. (implying
indistinguishability of the wealth processes $X^{H^i,c^i,e^i}$ and
$X^{\tilde{H}^i,c^i,e^i}$)
 and
 \item $\sum_{i} H^i_j(t)=0$, for all $t$, a.s. and all $j=1,\dots, n$.
 \end{enumerate}
 Therefore, the $d$-tuple $(H^i,c^i)$ satisfies the part 1. of
 Definition \ref{defequ}.

 It remains to show that $c^i$ maximizes $\UU^i$ over all
 consumption process $c'$ with
 $\UU^i(c')\in (-\infty,\infty)$ for which there exists a
 portfolio process $H'$ such that $(H',c')$ is
 $(S,B,e^i,\Gamma^i)$-affordable.
 We first note that each such $c'$ satisfies
 $\scl{\QQ}{c'}\leq \scl{\QQ}{e^i}$.
 This is due to (\ref{defhq}) and the fact that the discounted wealth
 $X'=\beta X^{H',c',e^i}$
 (which satisfies $X'_T\geq 0$) can
 be represented as a sum of a $\hQQ$-martingale and a term of the
 form $\int_0^t \beta_u (c'(u)-e^i(u))\, d\kappa_u$.
 Finally, because $c'\wedge k \in B^i(\QQ)$, for any $k\in\N$,
 the properties of the abstract equilibrium imply that
 $\UU^i(c^i)\geq \UU^i(c'\wedge k)$ and
 the Monotone Convergence Theorem yields $\UU^i(c^i)\geq \UU^i(c')$.
\end{proof}
\begin{theorem}\label{main} Suppose that
\begin{enumerate}
    \item $(\Omega, \FF, \prf{\FF},\PP)$ is a filtered
    probability space satisfying
Assumption~\ref{ftp},
    \item $(e^i)_{i=1,\dots,d}$ are random endowment
    processes verifying Assumption~\ref{ends},
    \item $(U^i)_{i=1,\dots,d}$ are utility functions for
    which Assumption \ref{utilreg} is valid, and
    \item $(\Gamma^i)_{i=1,\dots,d}$ are withdrawal cap
    processes satisfying Assumption \ref{with}.
\end{enumerate}
Then there exist an equilibrium market $(S,B)$ consisting of a
finite-dimensional semimartingale risky-asset process $S$ and a
positive predictable riskless-asset process $B$ of finite
variation for which the following additional properties hold
\begin{enumerate}
    \item The market $(S,B)$ is arbitrage free, i.e., there
    exists a unique measure $\hat{\QQ}$ equivalent to $\PP$, such that
    the discounted prices $S/B$ of risky assets are
    $\hat{\QQ}$-martingales.
    \item The optimal consumption densities $c^i$ in
    the market $(S,B)$ are uniformly bounded from above.
\end{enumerate}
\end{theorem}

\appendix
\section{Semimartingale functions and multiplicative
de\-com\-po\-si\-ti\-ons} \label{sec:semi} In this section we
provide several results which give sufficient conditions for 1) a
process obtained by applying a function to a semimartingale to be
a semimartingale, and 2) for a local martingale part in a
multiplicative decomposition of a positive process to be a
uniformly integrable martingale. These results can be improved in
several directions; we are aiming for conditions easily verifiable
in practice. In what follows, $I$ and $J$ will denote generic open
intervals in $\R$. For a process $A$ of finite variation,
$\abs{A}=\prf{\abs{A}_t}$ will denote its total variation process.
\paragraph{\em Semimartingale functions.}\label{sub:semi}
\begin{definition}
A function \label{defsemi} $f:[0,T]\times I\to\R$ is called a {\bf
semimartingale function} if the process $Y$ defined by
$Y_t=f(t,X_t)$, $t\in [0,T]$ is a \cd semimartingale for each
semimartingale $X$ taking values in $I$ and defined on an
arbitrary filtered probability space $(\Omega, \FF, \prf{\FF_t},
\PP)$.
\end{definition}

In this section we provide a set of sufficient conditions for a
function $f:[0,T]\times I\to\R$ to be a semimartingale function.
We go beyond  basic $C^{1,2}$-differentiability required by the
It\^ o formula and place much less restrictive assumptions on $f$.
Apart from being indispensable in Section \ref{sec:abstosto}, we
hope that the obtained result holds some independent probabilistic
interest.

\begin{theorem}
\label{semifun} Suppose that a function $f:[0,T]\times I\to\R$ can
be represented as $f(t,x)=f^1(t,x)-f^2(t,x)$, where for $i=1,2$,
\begin{enumerate}
\item $f^i$ is Lipschitz in the time variable, uniformly for $x$
in compact intervals. \item $f^i$ is convex in the second
variable. \item  The right derivative $f^i_{x+}$ is bounded on
compact subsets of $[0,T]\times I$ and satisfies
$f^i_{x+}(t,x)=\lim f^i_{x+}(s,x')$, when $(s,x')\to (t,x)$ and
$x'\geq x$.
\end{enumerate}
Then $f$ is a semimartingale function. Moreover, for a
semimartingale $X$ the local martingale part $\tilde{M}$ in the
semimartingale decomposition of
$f(t,X_t)=f(0,X_0)+\tilde{M}_t+\tilde{A}_t$ is given by
$\tilde{M}_t=\int_0^t f_{x+}(s,X_{s-})\, dM_s$, where $M$ is the
local martingale part in the semimartingale decomposition
$X_t=X_0+M_t+A_t$.
\end{theorem}

Before delving into the proof of Theorem \ref{semifun}, we recall
the concept of Fatou-convergence and some useful compactness-type
results related to it.

\begin{definition}
A sequence $(X^n)_{n\in\N}$ of \cd adapted processes is said to
{\bf Fatou-converge} towards a \cd adapted process $X$ if
\begin{equation}
    \nonumber
    \begin{split}
        X_t=\lim_{q\searrow t} \lim_n X^n_q,\, a.s.,
        \text{\ for all $t\in [0,T)$, and}\ \ X_T=\lim_n X^n_T,\,
         a.s,
    \end{split}
\end{equation}
where the first limit is taken over rational numbers $q>t$.
\end{definition}

\begin{lemma}\ \label{Fat}\begin{enumerate}
    \item Let $(A^n)_{n\in\N}$ be a sequence of non-decreasing
    adapted \cd processes
    taking values in $[0,\infty)$.
Then there exists a sequence $(\tilde{A}^n)_{n\in\N}$ of convex
combinations $\tilde{A}^n\in\conv(A^n, A^{n+1},\dots)$ and a
non-decreasing \cd process $\tilde{A}$ taking values in
$[0,\infty]$ such that $\tilde{A}^n$ Fatou-converges to
$\tilde{A}$.
    \item Let $(A^n)_{n\in\N}$ be a sequence of finite-variation
    \cd processes on $[0,T]$, with
    uniformly bounded total variations, i.e.,
\begin{equation}
    \nonumber
    \begin{split}
\text{$\abs{A^n}_T\leq C$ a.s., for some constant $C>0$ and all
$n\in\N$.}
    \end{split}
\end{equation}
 Then there exists a sequence
$(\tilde{A}^n)_{n\in\N}$ of convex combinations
$\tilde{A}^n\in\conv(A^n, A^{n+1},\dots)$ and a \cd process
$\tilde{A}$ of finite variation with $|\tilde{A}|_T\leq C$ such
that $\tilde{A}^n$ Fatou-converge towards $\tilde{A}$.
\end{enumerate}
\end{lemma}
\begin{proof} Part 1.  is a restatement of Theorem 4.2 in
\cite{Kra96a}. To prove part 2.,  note that the boundedness of
total variations of processes $A^n$ implies that the increasing
and decreasing parts $A^{\uparrow,n}$ and $A^{\downarrow,n}$ of
$A^n$ satisfy $A^{\uparrow,n}_T+A^{\downarrow,n}_T\leq C$ a.s. for
all $n$. Applying part 1. to increasing and decreasing parts and
noting that the limiting processes $\tilde{A}^{\uparrow}$ and
$\tilde{A}^{\downarrow}$ satisfy
$\tilde{A}^{\uparrow}+\tilde{A}^{\downarrow}\leq C$ a.s., leads to
the desired conclusion.
\end{proof}
\begin{proof}[Of Theorem \ref{semifun}.]
Let $X$ be a semimartingale taking values in the open interval
$I$. Our goal is to prove that the process $Y$ defined by
$Y_t=f(t,X_t)$ is a semimartingale. We first extend the
time-domain of $X$ and $Y$ by setting $X_t=X_T$ and $Y_t=f(t,X_T)$
for $t\in (T,\infty)$. By Theorem 6, p. 54 in \cite{Pro04}, it
will be enough to find an increasing sequence $(T_n)_{n\in\N}$ of
stopping times with $T_n\nearrow \infty$, a.s.,
 such that  the {\em pre-stopped} processes $Y^{T_n-}$ defined by
\begin{equation}
    \nonumber
    \begin{split}
  Y^{T_n-}_t= Y_t\inds{0\leq t<T_n}+Y_{T_n-}\inds{t\geq T_n}=f(t\wedge T_n,X^{T_n-}_t)
    \end{split}
\end{equation}
are semimartingales. Taking $T_n=\inf\sets{t\geq 0}{X_t\geq
n}\wedge n$, we reduce the problem to the case where the
semimartingale $X$ takes values in a compact interval $[x_1,x_2]$,
for $t\in [0,S)$, where $S=T\wedge T_n$.

Let $\eta^n: \R\times \R\to\R$ be a sequence of standard mollifier
functions with supports lying in the lower half-plane and
shrinking to a point, i.e.,
\begin{enumerate}
    \item $\eta^n\in C^{\infty}(\R\times \R)$.
    \item $\eta^n(t,x)\geq 0$, for all $t,x$ and
    $\int_{\R\times\R} \eta^n(t,x)\, dt\, dx=1$.
    \item The supports $\SS_n$ of $\eta^n$ satisfy
    $\SS_n\subseteq \R\times (-\infty,0]$ and
    $\abs{t}+\abs{x}\leq 1/n$ for all $(t,x)\in \SS_n$.
\end{enumerate}
Let the functions $f^n:[0,T]\times I_n\to\R$, where
$I_n=\sets{x\in I}{ d(x,I^c)>1/n}$,
 be the mollified versions of $f$, i.e.,
\begin{equation}
   \nonumber
   \begin{split}
 f^n(t,x)=(\eta^n * f)(t,x)=\int_{\R\times\R} \eta^n(s,y) f(t-s, x-y)\, ds\, dy,
   \end{split}
\end{equation}
where we set $f(t,x)=f(T,x)$ for $t>T$ and $f(t,x)=f(0,x)$ for
$t<0$. By standard arguments, the functions $f^n(t,x)$ have the
following properties
\begin{enumerate}
   \item $f^n(t,x)\to f(t,x)$ for all $(t,x)\in [0,T]\times I$,
   uniformly on compacts.
   \item $f^n(t,x)\in C^{\infty}([0,T]\times I_n)$.
   \item Let $C>0$ be a constant such that
   $\abs{f(t_2,x)-f(t_1,x)}\leq C\abs{t_2-t_1}$, for all
   $t_1,t_2\in [0,T]$ and $x\in [x_1,x_2]$.
   Then the absolute value $|f^n_t|$ of the time derivative $f_t^n$
   is bounded by the constant $C$, uniformly over $n\in\N$ and
   $(t,x)\in [0,T]\times [x_1,x_2]$.
   \item By condition 3. in the statement of the theorem and
   the fact that
   the support $\SS_n$ lies in the lower half-plane, we have
   $f^n_{x}(t,x)\to f_{x+}(t,x)$, for all $(t,x)\in [0,T]\times I$.
   \end{enumerate}

For $n\in\N$ such that $[x_1,x_2]\subseteq I_n$, the It\^ o
formula applied to $f^n$ implies that
$f^n(t,X_t)=f^n(0,X_0)+M^n_t+ A^n_t+B^n_t$, where
\begin{equation}
   \nonumber
   \begin{split}
   M^n_t&= \int_0^t f^n_x(s,X_{s-})\, dX_s,\quad
 A^n_t  = \int_0^t f^n_t(s,X_s)\, ds,\ \text{and}\\
 B^n_t & = \frac{1}{2} \int_0^t f^n_{xx}(x,X_{s-})\, d[X,X]^c_s\\
 &\qquad +
 \sum_{0< s\leq t}
 \big( f^n(s,X_s)-f^n(s,X_{s-})-f^n_x(s,X_{s-})\Delta X_s \big).
   \end{split}
\end{equation}
Note that:
\begin{enumerate}
    \item Using properties 3. and 4. (above) of $f^n$ and the
    Dominated Convergence Theorem for stochastic
integrals (see \cite{Pro04}, Theorem 32, p. 174), we have
$M^n_t\to M_t=\int_0^t f_{x+}(s,X_{s-})\, dX_s$, uniformly in
$t\in [0,T]$, in probability. It suffices to take a subsequence to
obtain convergence in the Fatou sense. \item By convexity of $f^n$
in the second variable, the processes $B^n_t$ are non-decreasing.
Thus, by Lemma \ref{Fat}, after a passage to a sequence of convex
combinations they Fatou-converge towards a non-decreasing \cd
adapted process $B$ taking values in $[0,\infty]$. \item The
processes in the sequence $A^n_t$ have total variation uniformly
bounded by $CT$, so by  part 2. of Lemma \ref{Fat}, there exists a
sequence of their convex combinations Fatou-converging towards a
process $A$ of finite variation with the total variation bounded
by the same constant $CT$.
\end{enumerate}
Compounding all subsequences and sequences of convex combinations
above, we obtain that $f(t,X_t)-f(0,X_0)-M_t-A_t=B_t$. We can
conclude that $B_T<\infty$, a.s. and
$f(t,X_t)=f(0,X_0)+M_t+A_t+B_t$ is a semimartingale.
\end{proof}
\begin{proposition} Every locally convexity-Lipschitz function
$f:[0,T]\times I\to\R$  admits a decomposition\label{regissemi}
$f=f^1-f^2$, where $f^1$ and $f^2$ satisfy  conditions 1.-3. of
Theorem \ref{semifun}. In particular, $f$ is a semimartingale
function.
\end{proposition}
\begin{proof} We shall construct the desired decomposition
only on a compact interval $[x_1,x_2]$ in $I$, as the general case
follows immediately.

For a fixed $t\in [0,T]$, the finite-variation function
$f_{x}(t,\cdot)$ admits a decomposition into a difference of a
pair $f^{\uparrow}(t,\cdot)$ and $f^{\downarrow}(t,\cdot)$ of
non-increasing and non-negative functions. Lipschitz continuity of
the total variation of the derivative $f_x$ implies that the
functions $f^{\uparrow}$ and $f^{\downarrow}$ are Lipschitz
continuous in $t$, uniformly in $x\in [x_1,x_2]$. It is now easy
to check that the sought-for decomposition is $f=f^1-f^2$, where
$f^1(t,x)=f(t,x_1)+\int_{x_1}^x f^{\uparrow}(t,\xi)\, d\xi$, and
$f^2(t,x)=\int_{x_1}^x f^{\downarrow}(t,\xi)\, d\xi$.
\end{proof}
\begin{proposition}\label{invreg}
Let $f:[0,T]\times I\to\R$ be locally convexity-Lipschitz, with
the derivative $f_x$ positive and bounded away from $0$ on compact
subsets of $[0,T]\times I$. If the function $g:[0,T]\times J\to\R$
satisfies $f(t,g(t,y))=y$, for all $(t,y)\in [0,T]\times J$,
 then $g$ is a semimartingale function.
\end{proposition}
\begin{proof}
We note first that the assumptions of the proposition imply that
both $f$ and $g$ are continuous and strictly increasing in the
second argument. To simplify the proof, we shall restrict the
domain of $g$ to a compact set of the form $[0,T]\times
[y_1,y_2]$, so that the range of $g$ is contained in a compact set
$[x_1,x_2]\subseteq I$. The general case will follow by {\em
pre-stopping} - the technique used in the proof of Theorem
\ref{semifun}.  Using the relationships
$0=f(t,g(t,y))-f(s,g(s,y))$ and $g_y(t,y)f_x(t,g(t,y))=1$ together
with the properties of function $f$ postulated in the statement,
it is tedious but straightforward to prove that both $g$ and $g_y$
are Lipschitz continuous in both variables.

Our next task is to decompose the function $g$ into a difference
of two functions satisfying conditions 1.-3. in Theorem
\ref{semifun}. By Proposition \ref{regissemi}, $f$ has a
decomposition $f(t,x)=f^1(t,x)-f^2(t,x)$ with properties 1.-3.
from Theorem \ref{semifun}. Let $h^i(t,y)$ denote the compositions
$f^i_{x}(t,g(t,y))$, $i=1,2$, and let $h(t,x)=h^1(t,x)-h^2(t,x)$
so that $f_{x}(t,g(t,y))=h(t,y)$. Then, for $i=1,2$,
$h^i(t,\cdot)$ is a non-decreasing function and for $y\in
[y_1,y_2]$,
\begin{equation}
    \label{derg}
    \begin{split}
    g_{y}(t,y)-g_{y}(t,y_1)&=
    - \int_{y_1}^y \frac{h(t,d\eta)}{(f_{x}(t,g(t,\eta)))^2} \\
  &= -
  \int_{y_1}^y g_{y}(t,\eta)^2 \, (h^1(t,d\eta)-h^2(t,d\eta)),
    \end{split}
\end{equation}
where $h^i(t,d\eta)$ stands for the Lebesgue-Stieltjes measure
induced by $h^i(t,\cdot)$. With $g^1$ and $g^2$ defined as
\begin{equation}
    \nonumber
    \begin{split}
  g^1(t,y)& =g(t,y_1)+g_{y}(t,y_1)(y-y_1)+
  \int_{y_1}^y \int_{y_1}^z g_{y}(t,\eta)^2 \, h^2(t,d\eta)\, dz,\\
  g^2(t,y)&=
  \int_{y_1}^y \int_{y_1}^z g_{y}(t,\eta)^2 \, h^1(t,d\eta)\, dz,
    \end{split}
\end{equation}
(\ref{derg}) implies that $g(t,y)=g^1(t,y)-g^2(t,y)$.

What follows is the proof of Lipschitz-continuity of $g^2_y(\cdot,
y)$. A simple change of variables - valid due to the continuity of
the function $g$ - yields $g^2_y(t,y)=\int_{g(t,y_1)}^{g(t,y)}
g_y(t,f(t,\xi))^2 f^1_x(t,d\xi)$, so, for $t,s\in [0,T]$ the
difference $g^2_y(t,y)-g^2_y(s,y)$ can be decomposed into the sum
$I_1+I_2+I_3+I_4$ where
      \[I_1 =\int_{g(t,y_1)}^{g(s,y_1)} g_y(s,f(s,\xi))^2 f^1_x(s,d\xi)\,
      ,I_2 =\int_{g(s,y)}^{g(t,y)} g_y(s,f(s,\xi))^2 f^1_x(s,d\xi),\]
\begin{equation}
    \nonumber
    \begin{split}
      I_3& =\int_{g(t,y_1)}^{g(t,y)}
      \big( g_y(t,f(t,\xi))^2-g_y(s,f(s,\xi))^2\big) f^1_x(t,d\xi),
      \ \text{and}
    \end{split}
\end{equation}
\begin{equation}
    \nonumber
    \begin{split}
      I_4& =\int_{g(t,y_1)}^{g(t,y)} g_y(s,f(s,\xi))^2
      \big(f^1_x(t,d\xi)-f^1_x(s,d\xi)\big)
    \end{split}
\end{equation}
Due to boundedness of $g$ and $g_y$ and Lipschitz continuity of
$g,g_y$ and $f_x$, the absolute values of the expressions $I_1$,
$I_2$ and $I_3$ are easily seen to be bounded by a constant
multiple of  $\abs{t-s}$. Additionally, the Lipschitz property of
the total-variation functional allows us to conclude the same for
$I_4$. Consequently, there exists a constant $C$ such that
$\abs{g^2_y(t,y)-g^2_y(s,y)}\leq C \abs{t-s}$, for all $y\in
[y_1,y_2]$.

Finally, to show that $g$ is a semimartingale function,  it
suffices to check that both $g^1$ and $g^2$ satisfy conditions
1.-3. of Theorem \ref{semifun}. The increase of the functions
$h^1(t,\cdot)$ and $h^2(t,\cdot)$ implies that $g^1(t,\cdot)$ and
$g^2(t,\cdot)$ are convex. Lipschitz-continuity of $g$ and $g^2$
in the time variable implies the same for $g^1=g+g^2$. Finally,
the derivatives $g^1_y$ and $g^2_y$ are continuous due to the
continuity of  functions $(f^1)_x(t,\cdot)$ and
$(f^2)_x(t,\cdot)$.
\end{proof}

\paragraph{\em The multiplicative decomposition of positive semimartingales.}
\label{sub:mult} A key step in the transition from abstract to
stochastic equilibria is the multiplicative decomposition of the
pricing functional which enforces the abstract equilibrium. In
this paragraph we give sufficient conditions on a positive
semimartingale in order for the local martingale part in its
multiplicative decomposition to be, in fact, a uniformly
integrable martingale.

The following proposition establishes some useful stability
properties of the condition $\NN(X)\in\linf$.
\begin{proposition}\ \label{stab}
\begin{enumerate}
    \item Let $X^1$ and $X^2$ be
semimartingales, and let $X=\min(X^1,X^2)$.  If $\NN(X^1)\in\linf$
and $\NN(X^2)\in \linf$, then $\NN(X)\in \linf$.
    \item Suppose $f:[0,T]\times I\to\R$ is a function  verifying the
conditions of Theorem \ref{semifun}, and $X$ is a bounded positive
semimartingale, bounded away from $0$, such that $\NN(X)\in\linf$.
Then the process $Y$, defined by $Y_t=f(t,X_t)$, satisfies
$\NN(Y)\in \linf $.
\end{enumerate}
\end{proposition}
\begin{proof}\
\begin{enumerate}
\item Let $X=M+A$, $X^{i}=M^{i}+A^{i}$, $i=1,2$ be the
semimartingale decompositions of $X$, $X^1$ and $X^2$. The
Meyer-It\^ o formula (see Theorem 70, p. 214 in \cite{Pro04})
states that $M_t=\int_0^t \inds{X^1_{s-}\leq X^2_{s-}} dM^1_s+
\int_0^t \inds{X^1_{s-}> X^2_{s-}} dM^2_s$, so $\scl{M}{M}_T\leq
\scl{M^1}{M^1}_T+\scl{M^2}{M^2}_T\in\linf$.

\item Assume that $X$ takes values in $[\eps,1/\eps]$, for some
$\eps>0$. It suffices to note that the right-continuous function
$f_{x+} (t,x)$ is bounded  on the compact set $[\eps,1/\eps]\times
[0,T]$, and that Proposition \ref{semifun} implies that the local
martingale part of the semimartingale $Y_t$ is given by $\int_0^t
f_{x+}(s,X_{s-})\, dM_s$.
\end{enumerate}
\end{proof}

\begin{proposition}\label{Decomp}
Let $X$ be a positive semimartingale bounded from above and away
from zero, such that $\NN(X)\in \linf$. Then $X$ admits a
multiplicative decomposition $X=\hQ\beta$ where $\beta$ is a
positive predictable process of finite variation, and $\hQ$ is a
positive uniformly integrable martingale.
\end{proposition}
\begin{proof}
Without loss of generality we assume $X_0=1$. By Theorem 8.21, p.
138 in \cite{JacShi03}, along with the semimartingale
decomposition $X=M+A$, $X$ also admits a multiplicative
decomposition of the form $X=\hQ\beta$. The same theorem states
that $\hQ=\EN(\hat{M})$ and $\ 1/\beta=\EN(\hat{A})$, where
$\hat{M}_t=\int_0^t H_s\, dM_s$, $\hat{A}_t=-\int_0^t H_s\, dA_s$,
and $H=\frac{1}{X_{-}+\Delta A}$ is the reciprocal of the
predictable projection ${}^p(X)$ of $X$. For a constant $\eps>0$
such that $\eps\leq X\leq 1/\eps$ we obviously have $\eps\leq
H\leq 1/\eps$, a.s. Thanks to the boundedness of $H$, the
compensator $\scl{\hat{M}}{\hat{M}}$ of the quadratic variation
$[\hat{M},\hat{M}]$ satisfies $\scl{\hat{M}}{\hat{M}}_T=\int_0^T
H_u^2 \, d\scl{M}{M}_u\leq 1/\eps^2 \scl{M}{M}_T=\NN(X)\in\linf$.
The conclusion now follows from Th{\' e}or{\` e}me 1, p. 147 in
\cite{MemShi79}, aided by the fact that absolute values
$\abs{\Delta M}$ of the jumps of the local martingale $M$ are
uniformly bounded (see Lemma 4.24, p. 44 in \cite{JacShi03}).
\end{proof}

\def\cprime{$'$} \def\cprime{$'$}
\providecommand{\bysame}{\leavevmode\hbox to3em{\hrulefill}\thinspace}
\providecommand{\MR}{\relax\ifhmode\unskip\space\fi MR }
\providecommand{\MRhref}[2]{%
  \href{http://www.ams.org/mathscinet-getitem?mr=#1}{#2}
}
\providecommand{\href}[2]{#2}

\end{document}